\documentclass[11pt]{article}

\usepackage[in]{fullpage}
\usepackage[utf8]{inputenc}
\usepackage{amsfonts,amssymb,amsthm,bbm, mathtools}
\usepackage{enumerate}
\usepackage{blkarray}
\usepackage{hyperref}
\usepackage{authblk}
\usepackage{footnote}
\usepackage{marginnote}
\usepackage{tablefootnote}
\usepackage{scrextend}

\hypersetup{
	colorlinks = true,
	citecolor = blue
}

\newtheorem{theorem}{Theorem}

\newtheorem{corollary}[theorem]{Corollary}

\theoremstyle{definition}
\newtheorem{definition}{Definition}

\newcommand{\matindex}[1]{\mbox{\scriptsize#1}}
\DeclareMathOperator{\transpose}{T}

\DeclareMathOperator{\e}{\mathbf e}

\DeclareMathOperator{\MV}{MV}

\DeclareMathOperator{\bin}{bin}
\DeclareMathOperator{\DI}{DI}
\DeclareMathOperator{\RI}{RI}
\DeclareMathOperator{\rank}{rank}

\newcommand{\ip}[1]{\langle #1 \rangle}
\newcommand{\clX}{\mathcal X}
\newcommand{\clY}{\mathcal Y}
\newcommand{\clU}{\mathcal U}
\newcommand{\clM}{\mathcal M}

\newcommand{\vx}{\vec\vecx}
\newcommand{\vy}{\vec\vecy}

\newif\ifnotes
\notestrue

\ifnotes
\newcommand{\notes}[1]{#1}
\else
\newcommand{\notes}[1]{}
\fi

\newcommand{\hide}[1]{}

\newcommand{\cX}{\mathcal{X}}
\newcommand{\cY}{\mathcal{Y}}

\newcommand{\Z}{\mathbb{Z}}

\newcommand{\vecx}{\mathbf x}
\newcommand{\vecy}{\mathbf y}

\newcommand{\zo}{\{0,1\}}

\providecommand{\fn}[1]{\ifmmode\operatorname{\mbox{\textnormal{#1}}}\else\mbox{\textnormal{#1}}\fi}
\newcommand{\sffn}[1]{\ifmmode\operatorname{\mbox{\textnormal{\textsf{#1}}}}\else\mbox{\textnormal{\textsf{#1}}}\fi}

\providecommand{\qedhere}{
\ifmmode
  \eqno \def\@badmath{$$}
    \let\eqno\relax \let\leqno\relax \let\veqno\relax
    \hbox{\qed}
\else
  \qed
\fi
}

\newcommand{\Pdisj}{\mathbf{DISJ}}

\newcommand{\Pgt}{\mathbf{GT}}
\newcommand{\Pleq}{\mathbf{LEQ}}
\newcommand{\Peq}{\mathbf{EQ}}
\newcommand{\Pneq}{\mathbf{NEQ}}
\newcommand{\Pindex}{\mathbf{INDEX}}
\newcommand{\Ppoly}{\mathbf{MPOLY}}

\newcommand{\Pthr}{\mathbf{THR}}
\newcommand{\Pethr}{\mathbf{ETHR}}
\newcommand{\Poreq}{\mathbf{OR{-}EQ}}

\title{On the Inner Product Predicate and \\ a Generalization of Matching Vector Families}

\date{}

\begin{document}

\renewcommand\Affilfont{\small}

\author[1]{Balthazar Bauer}
\author[2]{Jevgēnijs Vihrovs \thanks{Work done in part while interning at Centre for Quantum Technologies at NUS. Supported by the ERC Advanced Grant MQC.}}
\author[3]{Hoeteck Wee \thanks{Work done in part while visiting Centre for Quantum Technologies at NUS. Supported in part by ERC Project aSCEND (H2020 639554).}}
\affil[1]{ENS, Paris\\ 45 Rue d'Ulm, 75005 Paris, France}
\affil[2]{Centre for Quantum Computer Science, Faculty of Computing,\protect\\ University of Latvia, Raiņa 19, Riga, Latvia, LV-1586}
\affil[3]{CNRS and ENS, Paris, 45 Rue d'Ulm, 75005 Paris, France}

\maketitle

\begin{abstract}
  Motivated by cryptographic applications such as predicate
  encryption, we consider the problem of representing an arbitrary
  predicate as the inner product predicate on two vectors. Concretely, fix
  a Boolean function $P$ and some modulus $q$. We are interested in
  encoding $x$ to $\vx$ and $y$ to $\vy$ so that
  \[P(x,y) = 1 \Longleftrightarrow \ip{\vx,\vy}= 0 \bmod q,\]
  where the vectors should be as short as possible.
  This problem can also be viewed as a generalization of
  matching vector families, which corresponds to the equality
  predicate. Matching vector families have been used in the
  constructions of Ramsey graphs, private information
  retrieval (PIR) protocols, and more recently, secret sharing.

 Our main result is a simple lower bound that allows us to show that
  known encodings for many predicates considered in the
  cryptographic literature such as greater than and threshold are
  essentially optimal for prime modulus $q$.
  Using this approach, we also prove lower bounds on encodings for composite $q$,
  and then show tight upper bounds for such predicates as greater than, index and disjointness.
  
\end{abstract}

\section{Introduction}

There are many situations in cryptography where one is interested in
computing some function $F$ of a sensitive input $x$ but the
computational model is restricted so that only ``simple'' functions
$F$ can be directly computed.  For instance, the entries of $x$ may be
encrypted so that only {\em affine} functions can be computed, or
distributed between multiple non-interacting parties so that only {\em
  local} functions can be computed, or simply that we only know how to
construct schemes for handling simple functions.

For all of these reasons, it is useful to be able to ``encode''
complex functions as simple functions. An extremely influential
example of an ``encoding'' in the cryptographic literature is that of
garbling schemes (or randomized encodings), which have found
applications in many areas of cryptography and elsewhere
(see~\cite{Yao82,FKN94,IK00,AIK04,App11,BHR12,PSbook} and references
therein).

In this work, we consider the problem of \emph{inner product 
  encoding}, namely, representing an arbitrary predicate as the inner
product predicate on two vectors. Concretely, fix a Boolean function $P$ (a predicate) and
some modulus $q$ (may be composite as well as prime). We are interested in mappings $x \mapsto \vx, y \mapsto \vy$ that map to vectors in $\mathbb Z_q^\ell$ such that for
all $x,y$:
  \[P(x,y) = 1 \Longleftrightarrow \ip{\vx,\vy}= 0 \bmod q,\]
and $\ell$ is as small as possible.
This notion is motivated by the study of predicate encryption in
\cite{KSW08}, where $q$ is typically very large, for instance, as
large as the domains of $P$, and can also be viewed as a natural
generalization of matching vector families to arbitrary predicates.

As an example, consider the equality
predicate over $[n]$. Here, if $q=2$, then it is not difficult
to show that the vectors must have length $\Omega(n)$. On the other
hand, if $q > n$, then it is sufficient to use vectors of length
$2$: the inner product of $(1,x)$ and $(y,-1)$ is $0 \bmod q$ iff $x=y$.
More generally, for any predicate $P :
\mathcal X \times \mathcal Y \rightarrow \{0,1\}$ and any prime $q \geq 2$,
the ``truth table'' construction achieves vectors of length
$\min\{ |\mathcal X|, |\mathcal Y|\}$.

Interestingly, inner product predicate encoding for the equality
predicate have been studied in combinatorics and complexity theory,
where they are known as matching vector families.  Moreover, matching
vector families have found many applications, including the
construction of Ramsey graphs, private information retrieval (PIR)
protocols \cite{Grolmusz2000,Yek08,Efremenko12,DGY11,DG15}, and more recently,
secret-sharing schemes \cite{LVW17,LVW18,LV18}. Here,
prior works showed that if $q$ is a prime, then we must use vectors of
length $\Omega(n^{\frac{1}{q-1}})$ \cite{DGY11}.

\subsection{Our results}

Our main results are nearly tight bounds for many predicates considered in
the cryptographic literature such as greater than and threshold, for both
prime and composite modulus $q$. In particular, we have the following results for
prime modulus $q$:

\begin{itemize}
\item Greater than predicate for numbers in $[n]$ requires vectors of length
  $n$.  This rules out the possibility of deriving the predicate
  encryption for range queries with $O(\sqrt{n})$ ciphertext and
  secret key sizes in \cite{BW07} as a special case of inner product
  predicate encryption.
\item Threshold for $n$-bit strings and threshold $t$ requires vectors
  of length $2^{n-t+1}$. This rules out the possibility of constructing
  full-fledged functional encryption schemes by carrying out FHE
  decryption in the lattice-based predicate encryption of Gorbunov,
  Vaikuntanathan and Wee \cite{GVW15} using a pairing-based functional
  encryption scheme for the inner product predicate.
\end{itemize}

We then investigate encodings for composite $q$, specifically when $q$ is a product of $k$ distinct primes.
In many cases, a lower bound of $\ell/k$ for composite $q$ follows naturally if our method gives lower bound $\ell$ for prime $q$.
For predicates such as greater than, index and disjointness, we are able to show tight lower and upper bounds for both prime and composite $q$.
The full summary of upper and lower bounds is shown in Table \ref{tab:res}, and the listed predicates are described in Section \ref{sec:pre}.

Finally, we also consider probabilistic inner product predicate encoding.
For example, there is a probabilistic encoding of length $O((\log n)^2)$ for the greater than predicate for numbers in $[n]$, while any deterministic encoding must have length $\Omega(n)$, if $q$ is prime.

%
%

\begin{table}[h]
\begin{center} 
\begin{tabular}{cccccc}\hline
predicate & \multicolumn{2}{c}{$q$ prime} &  \multicolumn{2}{c}{$q$ product of $k$ primes}\\
& upper & lower & upper & lower\\\hline
$\Peq_n$\tablefootnote{Bounds from previous works, see Section \ref{sec:eq} for references.}  & $O(qn^{\frac{1}{q-1}})$ & $\Omega(n^{\frac{1}{q-1}})$ & $2^{\tilde O\left((\log n)^{1/k}\right)}$ & $\Omega(\log n)$ \\
$\Pgt_n$ &  $n$ & $n$ & $n/k$ & $n/k$\\
$\Pdisj_n$\tablefootnote{\label{1sttablefoot}For sufficiently large $q$.}, $\Pindex_n, \Pneq_n$ & $n$ & $n$ & $n/k$ & $n/k$ \\
$\Pethr_n^t$\footref{1sttablefoot} \tablefootnote{Assuming $t \leq n-2$, see Section \ref{sec:ethr}.}& $n+1$ & $n/2$ & $n+1$ & $n/2k$ \\
$\Ppoly_n^{d,q}$ & $n^d$ & $n^d$ & $n^d$ & $n^d/k$ \\
$\Pthr_n^t$ & $n^{n-t+1}$ & $2^{n-t+1}$ & $n^{n-t+1}$ & $2^{n-t+1}/k$\\ 
$\Poreq_n^q$ & $2^n$ & $2^n$ & $2^n$ & $2^n/k$
\end{tabular}
\end{center}
\caption{Summary of upper and lower bounds}\label{tab:res}
\end{table}

\paragraph*{Our lower bound technique.} Our lower bound technique
is remarkably simple. Suppose that $q$ is prime and we can represent a predicate $P:
\mathcal X \times \mathcal Y \rightarrow \{0,1\}$ as an inner product
predicate on vectors of length $r$ corresponding to
mappings $x \mapsto \vx, y \mapsto \vy$. Following \cite{BDL13}, we consider a matrix $F$ of dimensions $|\mathcal X|
\times |\mathcal Y|$ over $\Z_q$ whose $(x,y)$'th entry is $\ip{\vx,\vy} \bmod q$.
Then the matrix $F$ has rank at most $r$, because we can write $F$ as
the product of two matrices of dimensions $|\mathcal X| \times r$ and
$r \times |\mathcal Y|$. Concretely, $F = UV$ where the $x$'th row of $U$
is $\vx^{\transpose}$ and the $y$'th column of $V$ is $\vy$. This means that to
show a lower bound on $r$, it suffices to show that $F$ has large rank,
e.g.~by exhibiting a full rank submatrix.

As an example, consider the greater than predicate on $[n]$ for any prime
modulus
$q$. Then, the matrix
$F$ is an $n \times n$ upper triangular matrix where all the entries on and above
the diagonal are non-zero. This matrix has rank $n$, hence any correct construction must have dimension at least $n$.
Note that the above lower bound argument breaks down when $q$ is composite.
In fact, if $q = 2^n$, there is an encoding for greater than with dimension
$1$: take $x \mapsto 2^x, y \mapsto 2^{n-y}$. Correctness follows from
the fact that $2^x \cdot 2^{n-y} = 0 \bmod 2^n \Leftrightarrow x \geq y$, and the
construction extends also to the setting where $q$ is a product of $n$
distinct primes.

In order to extend our lower bounds to composite $q$ that is the
product of $k$ distinct primes, we observe that if $F \bmod q$ contains a
triangular submatrix of dimensions $\ell \times \ell$, then there
exists some prime factor $p$ of $q$ such that $F \bmod p$ contains a
 triangular submatrix of dimensions $\ell/k \times \ell/k$; this
follows from looking at the CRT decomposition of $q$ and a pigeonhole
argument. This simple observation allows us to translate many of our
lower bounds to the composite modulus setting, which we prove to be
essentially optimal via new upper bounds.

For instance, for the ``greater than'' predicate, we obtain a tight
bound of $n/k$ when $q$ is a product of $k$ distinct primes; this is
sharp contrast to standard matching vector families (i.e., the
equality predicate), where we have constructions of length
$2^{\tilde{O}((\log n )^{1/k})}$ when $q$ is a product of $k$ distinct
primes. For the upper bound, we begin with a construction of length $1$
for $k=n$ and then derive the more general construction by treating
the inputs as vectors of length $n$ and then dividing that into $n/k$
blocks each of length $k$.

Finally, we extend our results to the randomized setting. Here, we use a similar argument to show that the minimum size of a probabilistic inner product encoding is upper bounded by the probabilistic rank introduced by Alman and Williams \cite{AW17}.

\paragraph{Organization.}
The paper is organized as follows.
We describe our lower bound method in Section \ref{sec:main}.
The notation and predicates used throughout the rest of the paper are defined in Section \ref{sec:pre}.
In Section \ref{sec:det} we describe lower and upper bounds for these predicates.
Finally, we consider probabilistic encodings in Section \ref{sec:rand}.

\section{Main Theorem} \label{sec:main}

In this section we describe our lower bound technique.
Let $P : \mathcal X \times \mathcal Y \rightarrow \{0, 1\}$
be a predicate, and $q \geq 2$ be the integer modulus.  We say that a matrix $F
: \mathcal X \times \mathcal Y$ \emph{represents} $P$ modulo $q$ if for all $x \in
\mathcal X, y \in \mathcal Y$, we have $F_{x,y} = 0 \bmod q$ iff
$P(x,y) = 1$.

An \emph{inner product encoding} of $P$ of length $\ell$ is a pair of mappings from $\clX, \clY$ to $Z_q^\ell$ that map $x, y$ to $\vx, \vy$ in a way that the matrix $F : \mathcal X \times \mathcal Y$ defined by $F_{x,y} = \ip{\vx, \vy} \bmod q = (\sum_{i=1}^{\ell} \vx_i \cdot \vy_i) \bmod q$ represents $P$.
Denote the length of the shortest reduction from $P$ to inner product modulo $q$ by $\DI(P,q)$ (Deterministic Inner product).
Then we have the following simple and effective lower bound method.
\begin{theorem} \label{thm:main} For any predicate $P$ and any prime $q \geq 2$, we have
\begin{equation*}\DI(P, q) = \min_{F} \rank(F),\end{equation*}
where $F$ is any matrix that represents $P$ modulo $q$.
\end{theorem}

\begin{proof}
We show that if $P$ can be represented by a matrix $F$ modulo $q$, then the necessary and sufficient length of the encoding from $P$ to $F$ is exactly $\rank(F)$.
The decomposition rank definition states that the rank of an $m \times n$ matrix $F$ is the smallest integer $r$ such that $F$ can be factored as $F=UV$, where $U$ is an $m \times r$ matrix and $V$ is a $r \times n$ matrix.
Let $U_{x,*}$ be the row vector of $U$ that corresponds to $x \in \clX$ and $V_{*,y}$ be the column vector of $V$ that corresponds to $y \in \clY$.
Then the pair of mappings $x \mapsto U_{x,*}^{\transpose}$ and $y \mapsto V_{*,y}$ is a correct encoding of $P$, which is also the shortest possible for $F$.
\end{proof}

Therefore, to show a lower bound on the length of an encoding for $P$, it is sufficient to exhibit a set of rows $R$ and a set of columns $C$ such that for any matrix $F$ that represents $P$, the submatrix $F[C,R]$ is a full rank submatrix.
Typically we find a large full rank upper triangular submatrix and apply Theorem \ref{thm:main}.
Other times, we prove a lower bound for some predicate $Q$, and then prove that the same lower bound holds for $P$ by showing a predicate reduction from $Q$ to $P$ (see Section \ref{sec:red} for details).

For composite $q$, we have the following lower bound:
\begin{theorem} \label{thm:composite}
Let $q = p_1\cdots p_k$ be a product of $k$ distinct primes.
Let $P$ be a predicate such that every matrix $F$ that represents $P$ modulo $q$ is a triangular $n \times n$ matrix such that all numbers on the main diagonal are non-zero modulo $q$.
Then
$$\DI(P,q) \geq n/k.$$
\end{theorem}

\begin{proof}
Let $F$ represent $P$ modulo $q$.
Let $F^{(i)} = F \bmod p_i$ (all entries taken modulo $p_i$).
Since all entries on the main diagonal of $F$ are non-zero, there exists $i \in [k]$ such that there at least $n/k$ non-zero entries on the main diagonal of $F^{(i)}$ by pigeonhole principle.
As $F^{(i)}$ is also a triangular matrix, the rank of $F^{(i)}$ modulo $p_i$ is at least $n/k$.
By Theorem \ref{thm:main}, the length of any encoding from $P$ to $F^{(i)}$ modulo $p_i$ must be at least $n/k$, hence also $\DI(P,q) \geq n/k$.
\end{proof}


\section{Definitions and Predicates} \label{sec:pre}

In this section, first we describe some of the notation used throughout the paper.
Then we define the predicates examined in the paper, and define the predicate reduction.

\paragraph{Notation.}

We denote the set of all subsets of $[n]$ by $2^{[n]}$.
For a set $S \subseteq [n]$, define the characteristic vector $\chi(S) \in \zo^n$ by
$$\chi(S)_i = \begin{cases} 1, & \text{if $i \in S,$} \\ 0, & \text{otherwise.}\end{cases}$$
Conversely, for a vector $x \in \{0, 1\}^n$, let $\chi^{-1}(x)$ be the characteristic set of $x$.

For simplicity, denote the characteristic vector of $\{i\}$ by $\e_i$ (the length is usually inferred from the context).
The characteristic vectors of $\varnothing$ and $[n]$ are denoted by $0^n$ and $1^n$.
We denote the identity matrix of dimension $n$ by $I_n$, and all ones matrix by $J_n$.

For a truth expression $T$, we define $[T]$ to be 1 if $T$ is true, and 0 if $T$ is false.
For example, $[x=y] = 1$ iff $x=y$.

For a number $x \in [2^n]$, let $\bin(x) \in \zo^n$ be the binary representation of $x-1$.

\paragraph{Predicates.} We consider the predicates listed below.
\begin{itemize}
\item Equality: $\cX = \cY = [n]$ and
$\Peq_n(x,y) = [x = y].$

\item Greater than: $\cX = \cY = [n]$ and
$\Pgt_n(x,y) = [x > y].$

\item Inequality: $\cX = \cY = [n]$ and
$\Pneq_n(x,y) = [x \neq y].$

\item Index: $\cX = \zo^n, \cY = [n]$ and
$\Pindex_n(x, i) = [x_i = 0].$
Here, $x_i$ denotes the $i$'th coordinate of $x$.  Note that
we can also interpret $x$ as the characteristic vector of a subset
of $[n]$.
Because in our model $0 \bmod q$ corresponds to ``true'', we have defined the index to be true if the bit value in the corresponding position is 0.

\item Disjointness: $\clX = \clY = 2^{[n]}$ and
$\Pdisj_n(S,T) = [S \cap T = \varnothing].$

\item Exact threshold: $\clX = \clY = 2^{[n]}$ and
$\Pethr_n^t(S,T) = [|S \cap T| = t],$
where $t \in [n]$ is the threshold parameter.

\item Threshold: $\clX = \clY = 2^{[n]}$ and
$\Pthr_n^t(S,T) = [|S \cap T| \geq t],$
where $t \in [n]$ is the threshold parameter.

\item Multilinear polynomials: $\mathcal X = \mathbb Z_q^n$, $\mathcal Y \subseteq \{p \mid p \in \mathbb Z_q[x_1,\ldots,x_n], \deg(p) \leq d\}$,
the latter is the set of all multilinear polynomials of degree at most $d$. Then
$\Ppoly_n^{d,q}(x,p) = [p(x_1, \ldots, x_n) = 0 \bmod q].$

\item Disjunction of equality tests: $\clX = \clY = \mathbb Z_q^n$ and
$\Poreq_n^q(x,y) = [\bigvee_{i=1}^n x_i = y_i].$
\end{itemize}

\paragraph{Reductions} \label{sec:red}

We say that a predicate $P_1 : \clX_1 \times \clY_1 \to \zo$ can be \emph{reduced} to a predicate $P_2 : \clX_2 \times \clY_2 \to \zo$ if there exist two mappings $f : \clX_1 \to \clX_2$ and $g : \clY_1 \to \clY_2$ such that $P_2(f(x),g(y)) = P_1(x,y)$ for all $x \in \clX_1, y \in \clY_1$ (or mappings $f : \clX_1 \to \clY_2$ and $g : \clY_1 \to \clX_2$).
In that case we write $P_2 \Rightarrow P_1$.

For example, consider the following reductions:

\begin{itemize}
\item $\Pdisj_n \Rightarrow \Pindex_n \Rightarrow \Pneq_n$.

The reduction $\Pdisj_n \Rightarrow \Pindex_n$ holds since $\Pindex_n(x,i) = \Pdisj_n(\chi^{-1}(x),\{i\})$.
On the other hand, $\Pindex_n \Rightarrow \Pneq_n$, as $\Pneq_n(i,j) = \Pindex_n(\e_i,j)$. 
%
%

\item $\Pindex_n \Rightarrow \Pgt_n$.

As $\Pgt_n(x,y) = \Pindex_n(\chi([y]),x)$, the reduction follows.

\item Let $P : \clX \times \clY \to \{0, 1\}$ be any predicate.
Then $\Pindex_{\min\{|\clX|,|\clY|\}} \Rightarrow P$.

Let $T$ be the $\clX \times \clY$ truth table of $P$ defined by $T_{x,y} = P(x,y)$.
Then we have $P(x,y) = \Pindex_{|\clX|}(T_x,y)$ and $\Pindex_{|\clX|} \Rightarrow P$.
Similarly, we also have $\Pindex_{|\clY|} \Rightarrow P$.
\end{itemize}

Effectively, then an inner product encoding for $P_2$ implies an encoding for $P_1$ and a lower bound for $P_1$ implies a lower bound for $P_2$.
This makes it easier to prove upper and lower bounds.
For example, as later we prove that $\DI(\Pindex_n,q) = n$ for prime $q$ (see Section \ref{sec:index}), the last reduction implies that $\DI(P,q) \leq \min\{|\clX|,|\clY|\}$ for all predicates $P$.

If $q$ is a product of $k$ distinct primes, then $\DI(P,q) \leq \min\{|\clX|,|\clY|\}/k$ for the same reason.
Therefore, for any predicate, if $k = \min\{|\clX|,|\clY|\}$, there is an encoding of $\clX$ and $\clY$ simply to numbers modulo $q$.

\section{Deterministic Encodings} \label{sec:det}

In this section, we apply our technique to provide lower bounds on deterministic inner product encodings for many well-known predicates.
For each of them, first we discuss the encodings and then proceed to prove lower bounds.

\subsection{Equality} \label{sec:eq}


An encoding for $\Peq_n$ over $q$ is a matching family of vectors modulo $q$ \cite{DGY11}.
The maximum size of a matching family of vectors of length $\ell$ modulo $q$ is denoted by $\MV(q,\ell)$ and has been studied extensively.
Lower and upper bounds on $\MV(q,\ell)$ give upper and lower bounds on $\DI(\Peq_n,q)$, respectively (in the relevant literature, usually $q$ and $\ell$ are denoted by $m$ and  $n$, respectively).

For prime $q$, a tight $\DI(\Peq_n,q) = \Theta(qn^{\frac{1}{q-1}})$ bound is known \cite{DGY11}.
If $q$ is a product of $k$ primes, we have a $2^{\tilde O\left((\log n)^{1/k}\right)}$ upper bound from \cite{Grolmusz2000}.
For any composite $q$, we also have an $\Omega(\log n)$ lower bound from \cite{DH13}.

Here, first we show two simple upper bounds for $q =2$ and $q \geq n$.
Then we reprove the optimal lower bound for $q =2$ using our rank lower bound.


\paragraph*{Upper bounds.}
For $q = 2$, we construct an encoding of length $n$.
Let $\vx = \e_x$ and $\vy = 1^n - \e_y$.
Then $\ip{\vx,\vy} = \ip{\e_x,1^n} - \ip{\e_x,\e_y} = 1 - [x=y]$, thus it is a correct inner product encoding and $\DI(\Peq_n, 2) \leq n$.

Let $q$ be any integer such that $q \geq n$.
Let $\vx = (1, x)$ and $\vy = (y, -1)$.
Then $\ip{\vx, \vy} = y - x$, so it is 0 iff $x = y$.
Therefore, $\DI(\Peq_n, q) \leq 2$.

\paragraph*{Lower bound.}
We show a matching lower bound for case $q = 2$.
There is a unique matrix $F$ over $\Z_2$ that represents $\Peq_n$, namely $F_{x,y} = 0 \bmod q \Leftrightarrow x = y$.
Express $F = J_n - I_n$.
By sub-additivity of rank, we have $\rank(F) \geq \rank(I_n) - \rank(J_n) = n-1$.
Hence, by Theorem \ref{thm:main}, any inner product encoding of $\Peq_n$ modulo 2 requires vectors of length at least $n-1$, that is, $\DI(\Peq_n, 2) \geq n-1$.

\subsection{Index} \label{sec:index}

We prove that $\DI(\Pindex_n,q) = \lceil n/k \rceil$, for every $q$
that is a product of $k$ distinct primes.

For some $q$, the upper bound follows from $\Pdisj_n \Rightarrow \Pindex_n$ (see Section \ref{sec:disj}).
However, there is a much simpler encoding, which we present below. Moreover,
this upper bound holds for every $q$ that is the product of $k$ distinct primes.

\paragraph*{Upper bound.}
We begin with the warm-up for the special case $k=n$. Here,
consider
$$\vx = \prod_{i=1}^n p_i^{1-x_i},\qquad \vy = q/p_y.$$
Then $\ip{\vx, \vy} = 0 \bmod q$ iff $x_y = 0$.


Next, we consider general $k,n$.
Since $\Pindex_{\lceil n/k\rceil \cdot k} \Rightarrow \Pindex_n$, it is enough to construct an encoding for the case $k \mid n$.
The data is the string $x \in \{0, 1\}^n$, and the index is given by $y \in [n]$.
Encode $x$ as an $n/k \times k$ binary matrix $X_{i,j} = x_{(i-1) \cdot k + j}$, and $y$ as an $n/k \times k$ binary matrix $Y_{i,j} = [y = (i-1) \cdot k + j]$.

Now we construct the encoding.
$$
\vx_i = \prod_{j=1}^k p_j^{X_{i,j}},
\hspace{1cm}
\vy_i = \begin{cases}
q/p_j, &\text{if $Y_{i,j}=1$,}\\
0, & \text{otherwise.}
\end{cases}
$$

Now we analyze the correctness of the protocol.
Let $i, j$ be such that $Y_{i,j} = 1$.
Then $\ip{\vx,\vy} = \prod_{l=1}^k p_l^{X_{i,l}} \cdot (q/p_j)$.
\begin{itemize}
\item If $X_{i,j} = 1$, then $\ip{\vx,\vy} = 0 \bmod q$.
\item If $X_{i,j} = 0$, then $p_j \nmid \ip{\vx,\vy}$, hence $\ip{\vx,\vy} \neq 0 \bmod q$.
\end{itemize}

\paragraph*{Lower bound.}

The lower bound follows from $\Pindex_n \Rightarrow \Pneq_n$ (see Section \ref{sec:neq}).

\subsection{Inequality} \label{sec:neq}

We show that $\DI(\Pneq_n,q) = \lceil n/k \rceil$, for every $q$
that is the product of $k$ distinct primes.

\paragraph*{Upper bound.}

The upper bound follows from $\Pindex_n \Rightarrow \Pneq_n$ (see Section \ref{sec:index}).

\paragraph*{Lower bound.}

Any matrix that represents $\Pneq_n$ is a diagonal matrix with non-zero entries on the main diagonal.
By Theorem \ref{thm:composite}, it follows that $\DI(\Pneq_n,q) \geq n/k$.

\subsection{Greater Than} \label{sec:gt}

We show that $\DI(\Pgt_n,q) = \lceil n/k \rceil$, for every $q$
that is the product of $k$ distinct primes.

\paragraph*{Upper bound.}
%
%

The upper bound follows from $\Pindex_n \Rightarrow \Pgt_n$ (see Section \ref{sec:index}).

If $q$ is prime, the encoding simplifies to $\vx = \e_x$ and $\vy = \sum_{i=1}^y \e_i$.
If $k = n$, a different simple encoding is $\vx = \prod_{i=1}^{x-1} p_i$ and $\vy = \prod_{i=y+1}^{n} p_i$.

\paragraph*{Lower bound.}
Let $F$ be any matrix that represents $\Pgt_n$ modulo $q$.
Then all entries below the main diagonal are 0, while all entries on and above the main diagonal are non-zero, hence $F$ is a triangular matrix.
By Theorem \ref{thm:composite}, we conclude that $\DI(\Pgt_n,q) \geq n/k$.

\subsection{Disjointness} \label{sec:disj}

We prove that $\DI(\Pdisj_n,q) = \lceil n/k \rceil$ for an appropriate choice of $q$ that depends on $n$, and that $\DI(\Pdisj_n,q) \geq n/k$ if $q$ is any product of $k$ distinct primes.

\paragraph*{Upper bound.}

We start with a simple encoding for $k = n$ that works for any product of $n$ distinct primes $q$.
Recall that the sets $S$ and $T$ are the input to disjointness.
Let
$$ \vx = \prod_{i=1}^n p_i^{1-\chi(S)_i}, \hspace{1cm} \vy = \prod_{i=1}^n p_i^{1-\chi(T)_i}.$$
Then $\ip{\vx,\vy} = \prod_{i=1}^n p_i^{2-\chi(S)_i-\chi(T)_i}$ is $0 \bmod q$ iff $S$ and $T$ are disjoint.
If $k < n$, then for any $p_i$ it is possible that although some of the products $\vx_i \cdot \vy_i$ are not divisible by $p_i$, their sum might be divisible by $p_i$, hence the encoding doesn't work for any $q$.

For the general case, the following variation of Dirichlet's theorem will be useful for us.
\begin{theorem}[Dirichlet] \label{thm:prime}
For any integer $q \geq 2$, there are infinitely many primes $p$ such that $p = 1 \bmod q$.
\end{theorem}

Let $q = p_1\cdots p_k$ be a product of $k$ distinct primes $p_1, \ldots, p_k$ to be defined later.
We construct an encoding of length $n/k$ for the case $k \mid n$.
Encode $S \subseteq [n]$ as an $n/k \times k$ binary matrix $X_{i,j} = \chi(S)_{(i-1) \cdot k + j}$.
Similarly encode $T$ as $Y$.
Let
$$
\vx_i = \prod_{j=1}^k p_j^{1-X_{i,j}},
\hspace{1cm}
\vy_i = \prod_{j=1}^k p_j^{1-Y_{i,j}}.
$$

Now we find the appropriate primes $p_1, \ldots, p_k$ for the general case.
We construct them and prove the correctness by induction on $k$.

\textbf{Base case.}
If $k=1$, then $q$ is a prime itself.
Pick any prime $q$ such that $q > n$.
We have $\ip{\vx,\vy} = \sum_{i=1}^n q^{2-\chi(S)_i-\chi(T)_i}$.
If $x$ and $y$ are disjoint, then $q \mid \ip{\vx,\vy}$.
Suppose that $S$ and $T$ are not disjoint.
Let $b = |S \cap T|$.
Then $\ip{\vx,\vy} = (\sum_{i \in S \cap T} 1) \bmod q = b \bmod q$.
As $b \leq n$, we have $b < q$, therefore $\ip{\vx,\vy} \neq 0 \bmod q$.

\textbf{Inductive step.}
Assume that there exists a correct encoding for some $q$ such that it is a product of $k-1$ distinct primes $p_1, \ldots, p_{k-1}$.
Let $p_k$ be a prime such that $p_k > (n/k)\cdot (p_1\cdots p_{k-1})^2$ and $p_k = 1 \bmod (p_1\cdots p_{k-1})$ (such exist by Theorem \ref{thm:prime}).

Suppose that $\ip{\vx,\vy} = p_k\cdot a + b$, where $b \in \{0,\ldots,p_k-1\}$.
Examine the sets $S^{(k)}= \{i \in [n/k] \mid ik \in S\}$ and $T^{(k)}= \{i \in [n/k] \mid ik \in T\}$.
\begin{itemize}
\item Suppose that $S^{(k)}$ and $T^{(k)}$ are not disjoint.
Then the set $I = S^{(k)} \cap T^{(k)}$ is non-empty.
If $i \notin I$, then at least one of $X_{i,k}$ and $Y_{i,k}$ is 0, thus $\vx_i \cdot \vy_i = \prod_{j=1}^k p_j^{2-X_{i,j}-Y_{i,j}}$ is divisible by $p_k$.
Thus, we have that $b = \sum_{i \in I} \prod_{j=1}^{k-1} p_j^{2-X_{i,j}-Y_{i,j}} < (n/k) \cdot (p_1\cdots p_{k-1})^2 < p_k$.
Therefore, $\ip{\vx,\vy} = b \bmod p_k \neq 0 \bmod p_k$.
\item Suppose that $S^{(k)}$ and $T^{(k)}$ are disjoint.
Then for all $i \in [n/k]$, we have that $p_k \mid \vx_i\vy_i$.
Therefore, $p_k \mid \ip{\vx,\vy}$.

Moreover, since $p_k = 1 \bmod (p_1\cdots p_{k-1})$, we have that $\vx_i \bmod (p_1\cdots p_{k-1}) = \prod_{j=1}^{k-1} p_j^{1-X_{i,j}}$ and $\vy_i \bmod (p_1\cdots p_{k-1}) = \prod_{j=1}^{k-1} p_j^{1-Y_{i,j}}$.
Therefore, $\ip{\vx,\vy} \bmod (p_1\cdots p_{k-1})$ is equal to 0 iff the sets $S \setminus S^{(k)}$ and $T \setminus T^{(k)}$ are disjoint by the inductive hypothesis.
\end{itemize}

\paragraph*{Lower bound.}

The lower bound follows from $\Pdisj_n \Rightarrow \Pindex_n$ (see Section \ref{sec:index}).

\subsection{Exact Threshold} \label{sec:ethr}

\paragraph*{Upper bound.}

The following encoding modulo $q \geq n$ of length $n+1$ is due to Katz, Sahai and Waters \cite{KSW08}.
For all $1 \leq i \leq n$, let $\vx_i = \chi(S)_i$, and let $\vx_{n+1} = 1$.
For all $1 \leq i \leq n$, let $\vy_i = \chi(T)_i$, and let $\vy_{n+1} = -t$.
Then $\ip{\vx, \vy}$ is equal to 0 iff $|S \cap T| = t$.
Therefore, $\DI(\Pethr_n^t, q) \leq n+1$.

Surprisingly, if $t \geq n-1$, there exist constant size encodings.
\begin{itemize}
\item If $t = n$, there is an encoding of length 2.
The encoding is as follows: $\vx = (1, [S=[n]])$ and $\vy = (1,-[T=[n]])$.
Then we have $\ip{\vx, \vy} = 1 - [S=[n]]\cdot [T=[n]]$, which is 0 iff $S = T = [n]$.

\item If $t = n-1$, there is an encoding of length 3. 
The encoding for $S$ and $T$ is as follows:
\begin{equation*}
\vx = \begin{cases}
(1,0,0), & \text{if $|S|=n$,} \\
(0,i,1), & \text{if $|S|=[n]\setminus \{i\}$,} \\
(1,-1,1), & \text{otherwise.}
\end{cases}
\hspace{1cm}
\vy = \begin{cases}
(1,0,0), & \text{if $|T|=n$,} \\
(0,1,-i), & \text{if $|T|=[n]\setminus \{i\}$.} \\
(1,1,1), & \text{otherwise.}
\end{cases}
\end{equation*}
It is easy to check by hand that $\ip{\vx,\vy}=0$ iff $|S\cap T|=n-1$.
Note that we require $q \geq n+2$.
\end{itemize}

\paragraph*{Lower bound.}

We show that for $1 \leq t \leq n-2$, we have $\DI(P,q) \geq \max\{n-t+2, t+2\}/k \geq (n/2+2)/k$.

\begin{enumerate}[(a)]

\item 

First we prove that if $t \geq 1$, the length of any encoding must be at least $(n-t+2)/k$.
We show that by using two reductions.

Firstly, we have $\Pethr_n^t \Rightarrow \Pethr_{n-t+1}^1$, because we can map $S \mapsto S \cup \{n-t+2, \ldots, n\}$.
Secondly, we prove that $\Pethr_m^1 \Rightarrow \Pgt_{m+1}$.
Consider the following mappings:
\begin{equation}
f = \begin{cases}
1 \mapsto \varnothing, \\
i \mapsto [i-1],
\end{cases}
\hspace{1cm}
g = \begin{cases}
j \mapsto \{j\}, \\
m+1 \mapsto \varnothing.
\end{cases}
\end{equation}
Consider a pair of numbers $x, y \in [m+1]$.
If $x = 1$, then $\Pgt_{m+1}(x,y) = 0$ and also $\Pethr_m^1(f(x),g(y)) = \Pethr_m^1(\varnothing, g(y)) = 0$.
If $y = m+1$, then $\Pgt_{m+1}(x,y) = 0$ and $\Pethr_m^1(f(x),g(y)) = \Pethr_m^1(f(x), \varnothing) = 0$.
Otherwise, $\Pethr_m^1(f(x),g(y)) = \Pethr_m^1([x-1],\{y\}) = \Pgt_{m+1}(x,y)$.
Hence the reduction is correct.

Therefore, we conclude that $$\DI(\Pethr_n^t,q) \geq \DI(\Pethr_{n-t+1}^1,q) \geq \DI(\Pgt_{n-t+2},q) \geq (n-t+2)/k$$ by the lower bound on greater than of Section \ref{sec:gt}.

%
%
%

\item Now we prove that if $t \leq n-2$, the length of any encoding is at least $(t+2)/k$.
Again, we exhibit two reductions.

Firstly, $\Pethr_n^t \Rightarrow \Pethr_{t+2}^t$ simply mapping any set to itself.
Secondly, $\Pethr_m^{m-2} \Rightarrow \Pneq_m$.
This is because we can map $x \mapsto [m] \setminus \{x\}$ for any $x \in [m]$.
Then the size of the intersection $|([m]\setminus \{x\}) \cap ([m]\setminus \{y\})|$ is equal to $m-2$ if $x \neq y$, and $m-1$, if $x = y$.

Therefore, it follows that
$$\DI(\Pethr_n^t,q) \geq \DI(\Pethr_{t+2}^t,q) \geq \DI(\Pneq_{t+2},q) \geq (t+2)/k$$
by the lower bound on inequality of Section \ref{sec:neq}.

%
%
%
\end{enumerate}

Therefore, for any $1 \leq t \leq n-2$, any encoding must have length at least $\max\{n-t+2,t+2\}/k$ and we have that $\DI(\Pethr_n^t,q) = \Omega(n)$.

\subsection{Multilinear Polynomials} \label{sec:poly}

First we show a known encoding that gives $\DI(\Ppoly_n^d,q) \leq {n \choose \leq d} = O(n^d)$.
Then we show a lower bound of $\DI(\Ppoly_n^d,q) \geq {n \choose d}/ k = \Omega(n^d/k)$.
For prime $q$, we show an optimal lower bound $\DI(\Ppoly_n^d,q) \geq {n \choose \leq d}$.

\paragraph*{Upper bound.}

The following is a simple construction by \cite{KSW08}.
For $S \subseteq [n]$, let $X_S = \prod_{i\in S} x_i$ and let $p = \sum_{S \subseteq [n], |S| \leq d} a_S X_S$ be a multilinear polynomial of degree at most $d$.
For each subset $S \subseteq [n]$ such that $|S| \leq d$, let $\vx_S = X_S$ and $\vy_S = a_S$; then $\ip{\vx,\vy}$ is precisely equal to $p(x)$.
Since a multilinear polynomial of degree at most $d$ on $n$ variables has at most ${n \choose \leq d} = \sum_{i=0}^d {n \choose i} \leq (n+1)^d$ monomials, it follows that $\DI(\Ppoly_n^d,q) = O(n^d)$.

\paragraph*{Lower bound.}

We show a reduction $\Ppoly_n^d \Rightarrow \Pneq_{{n \choose d}}$.
Let $S$ be the bijection from the numbers in $\left[{n \choose d}\right]$ to subsets of $[n]$ of size $d$.
For a pair of inputs $x, y \in \left[{n \choose d}\right]$, consider mappings $x \mapsto \chi(S(x))$ and $y \mapsto X_{S(y)}$.
Since $\Ppoly_n^d(\chi(S(x)), X_{S(y)}) = 0$ iff $x \neq y$, it is a correct reduction.
Thus, $\DI(\Ppoly_n^d,q) \geq {n \choose d}/k = \Omega(n^d/k)$ by the lower bound from Section \ref{sec:neq}.

Note that if $q$ is prime, we can get a tight lower bound of ${n \choose \leq k}$.
Since any two distinct polynomials disagree on some inputs, each polynomial must be mapped to a different vector.
Therefore, the number of possible vectors must be at least the number of possible polynomials, $|\mathbb Z_q^n| \geq |\mathcal Y|$.
The total number of possible monomials of degree at most $d$ is ${n \choose \leq d}$.
Each monomial can have any coefficient in $\mathbb Z_q$.
It implies that $q^m \geq q^{{n \choose \leq d}}$ and $m \geq {n \choose \leq d}$.

%

\subsection{Threshold}

First we show an upper bound of $\DI(\Pthr_n^t,q) = O(n^{n-t+1})$ for $q \geq n$, and then a lower bound of $\DI(\Pthr_n^t,q) \geq 2^{n-t+1}/k$.

\paragraph*{Upper Bound.}

The idea is to encode the threshold into multilinear polynomial evaluation.
Let $x = \chi(S)$ and $y = \chi(T)$.
Examine the following polynomial:
\begin{equation*}p_y(x) = \left(\sum_{i=1}^n x_iy_i - t\right) \cdot \left(\sum_{i=1}^n x_iy_i - (t+1)\right)\cdot \ldots\cdot \left(\sum_{i=1}^n x_iy_i - n\right).\end{equation*}
Firstly, $\sum_{i=1}^n x_iy_i = |S \cap T|$, thus $p_y(x) = 0$ iff $|S \cap T| \geq t$.
Secondly, the degree of each factor is 1, hence $\deg(p_y) = n-t+1$.
Note that the polynomial $p_y$ is still multilinear, since all the variables are 0 or 1.
Therefore, we have a reduction $\Ppoly_n^{n-t+1} \Rightarrow \Pthr_n^t$.
The upper bound from Section \ref{sec:poly} implies that $\DI(\Pthr_n^t,q) \leq \DI(\Ppoly_n^{n-t+1},q) \leq {n \choose \leq n-t+1} = O(n^{n-t+1})$.


\paragraph*{Lower Bound.}

First of all, we have $\Pthr_n^t \Rightarrow \Pthr_{n-t+1}^1$, as we can map a set $S \subseteq [n-t+1]$ to $S \cup \{n-t+2,\ldots,n\}$.
Next we prove that $\DI(\Pthr_m^1,q) \geq 2^m/k$.

Let $F$ be any matrix representing $\Pthr_m^1$.
We show that $F$ is a triangular matrix with all entries on the main diagonal being non-zero.
Then the claim follows by Theorem \ref{thm:composite}.

Order the rows of $F$ by the increasing order of the size of the sets they correspond to.
Then order the columns of $F$ in such a way that the sets corresponding to the $i$-th row and the $i$-th column are the complements of each other.

As the complements don't overlap, the numbers on the main diagonal of $F$ are non-zero.
Now examine any entry on the $i$-th row and $j$-th column such that $i \geq j$.
Let $S$ correspond to the set of the $i$-th row and $T$ correspond to the set of the $j$-th column.
Since the columns are ordered by the decreasing size of the sets, we have that $|S| \geq m-|T|$, or equivalently $|S| + |T| \geq m$.

If $|S| + |T| > m$, then the sets must overlap and the value of $F_{i,j}$ is 0.
If $|S| + |T| = m$, then the only way $S$ and $T$ do not overlap is if $T$ is the complement of $S$.
In any case all the numbers below the main diagonal are 0, and non-zero on the main diagonal.
See Figure \ref{fig:matrix} for an example.

\begin{figure}[h]
\[
  \begin{blockarray}{ccccccccc}
    & \matindex{111} & \matindex{110} & \matindex{101} & \matindex{011} & \matindex{100} & \matindex{010} & \matindex{001} & \matindex{000} \\
    \begin{block}{c[cccccccc]}
       \matindex{000} & * & * & * & * & * & * & * & * \\
       \matindex{001} & 0 & * & 0 & 0 & * & * & 0 & * \\
       \matindex{010} & 0 & 0 & * & 0 & * & 0 & * & * \\
       \matindex{100} & 0 & 0 & 0 & * & 0 & * & * & * \\
       \matindex{011} & 0 & 0 & 0 & 0 & * & 0 & 0 & * \\
       \matindex{101} & 0 & 0 & 0 & 0 & 0 & * & 0 & * \\
       \matindex{110} & 0 & 0 & 0 & 0 & 0 & 0 & * & * \\
       \matindex{111} & 0 & 0 & 0 & 0 & 0 & 0 & 0 & *\\
    \end{block}
  \end{blockarray}
\]
\caption{An example of $F$ for $m = 3$. Stars represent non-zero elements.} \label{fig:matrix}
\end{figure}

\subsection{Disjunctions of Equality Tests}

We show that for prime $q$, we have $\DI(\Poreq_n^q,q) \leq 2^n$ and if $q$ is a product of $k$ distinct primes, then $\DI(\Poreq_n^q,q) \geq 2^n/k$.

\paragraph*{Upper bound.}
We prove that $\Ppoly_n^{n,q} \Rightarrow \Poreq_n^q$.
Examine a multilinear polynomial
\begin{equation*}
p_y(x) = \prod_{i=1}^n (x_i-y_i).
\end{equation*}
Clearly, $p_y(x) = 0 \bmod q$ iff at least one equality holds.
Therefore, if we map $x \mapsto x$ and $y \mapsto p_y$, then we have a correct reduction to multilinear polynomial evaluation.
By the upper bound from Section \ref{sec:poly}, we have $\DI(\Poreq_n^q,q) \leq \DI(\Ppoly_n^{n,q},q) \leq \sum_{i=0}^n {n \choose i} = 2^n$.

\paragraph*{Lower bound.}
We prove that $\Poreq_n^q \Rightarrow \Pneq_{2^n}$.
For the input $x, y \in [2^n]$ to $\Pneq_n$, map $x \mapsto \bin(x)$ and $y \mapsto \bin(y) \oplus 1^n$.
As $x \neq y$ iff there exists an $i$ such that $\bin(x)_i \neq \bin(y)_i$, we have that $x \neq y$ iff $\Poreq_n^q(\bin(x),\bin(y)\oplus 1^n) = 1$.
The lower bound follows by Section \ref{sec:neq}.

\section{Randomized Constructions} \label{sec:rand}

We can formulate the problem in the randomized setting as follows.
Let $P : \clX \times \clY \to \{0, 1\}$ be a predicate.
Consider all pairs of mappings $\clU = \{(x \mapsto \vx, y \mapsto \vy) \mid \vx, \vy \in \Z_q^\ell\text{ for some $\ell$}\}$.
These also include mappings that are incorrect inner product encodings of $P$.
Let $\mu$ be a probability distribution over $\clU$.
Then $\mu$ is a \emph{probabilistic inner product encoding} modulo $q$ with error $\epsilon$, if
\begin{equation*}
\Pr[P(x,y) \neq [\ip{\vx, \vy} = 0 \bmod q ] \mid (x \mapsto \vx, y \mapsto \vy) \sim \mu] \leq \epsilon.
\end{equation*}

We consider the length of the longest encoding under $\mu$ to be the length of $\mu$ and denote it by $\RI^{\mu}(P, q)$ (Randomized Inner product).
Then define
\begin{equation*}
\RI_{\epsilon}(P, q) = \min_\mu \RI^\mu (P, q),
\end{equation*}
where $\mu$ ranges over all probabilistic inner product encodings of $P$ modulo $q$ with error $\epsilon$.

Next we reproduce the definition of the probabilistic rank (over $\Z_q$) by Alman and Williams \cite{AW17}:
\begin{definition}[Probabilistic Matrix]
For $n, m \in \mathbb N$, define a \emph{probabilistic matrix} over $\Z_q$ to be a distribution of matrices $\clM \subset \Z_q^{n \times m}$.
A probabilistic matrix $\clM$ \emph{computes} a matrix $A \in \Z_q^{n \times m}$ with error $\epsilon > 0$ if for every entry $(i, j) \in [n] \times [m]$,
\begin{equation*}
\Pr_{M \sim \clM} [A_{i,j} \neq M_{i,j} ] \leq \epsilon.
\end{equation*}
\end{definition}
\begin{definition}[Probabilistic Rank]
Let $q$ be prime.
Then a probabilistic matrix $\clM$ has \emph{rank} $r$ if the maximum rank of an $M$ in support of $\clM$ is $r$.
Define the $\epsilon$-\emph{probabilistic rank} of a matrix $A \in \Z_q^{n \times m}$ to be the minimum rank of a probabilistic matrix computing $M$ with error $\epsilon$.
Denote it by $\rank_{\epsilon}(A)$.
\end{definition}

As we can see, the probabilistic choice of a distribution $\mu$ corresponds to a matrix $M$ sampled from $\clM$.
By a similar reasoning as in the proof of Theorem \ref{thm:main}, we have the following theorem:
\begin{theorem} \label{rand}
For any predicate $P$, prime $q \geq 2$ and error $\epsilon$, we have
\begin{equation*}
\RI_\epsilon(P, q) \leq \min_{F} \rank_\epsilon(F),
\end{equation*}
where $F$ is any matrix that represents $P$ modulo $q$.
\end{theorem}

\begin{proof}
Let $F$ be any matrix that represents $P$ modulo $q$.
Suppose that $\clM$ is a probabilistic matrix that computes $F$.
Then any $M$ in support of $\clM$ defines an encoding of length $\rank(M)$ by the decomposition rank.
Therefore, there is a probability distribution over the encodings such that the maximum length is $\rank_{\epsilon}(F)$.
\end{proof}

For some predicates, the probabilistic rank can be much smaller than the deterministic rank.
Let $T(P)$ be a truth table of a predicate $P$ (defined by $T(P)_{x,y} = P(x,y)$).
The same authors prove that $\rank_{\epsilon}(T(\Peq_n)) = O(1/\epsilon)$ and $\rank_{\epsilon}(T(\Pleq_n)) = O((\log n)^2/\epsilon)$ (see Lemmas D.1 and D.2 in \cite{AW17}).
Since the matrix $T(P)$ represents the predicate $\lnot P$ (in our setting), these results imply that for any prime $q$:
\begin{enumerate}
\item $\RI_{\epsilon}(\Pneq_n,q) = O(1/\epsilon)$,
\item $\RI_{\epsilon}(\Pgt_n,q) = O((\log n)^2/\epsilon)$.
\end{enumerate}

We conclude by showing that these results immediately imply a constant length probabilistic encoding for $\Peq_n$ modulo any prime:
\begin{corollary}
For any prime $q$, we have
$\RI_{\epsilon}(\Peq_n,q) = O(1/\epsilon).$
\end{corollary}

\begin{proof}
Let $\clM$ be a probabilistic matrix that computes $T(\Peq_n)$ with error $\epsilon$.
The matrix $F(\Peq_n) = J_n - T(\Peq_n)$ represents $\Peq_n$.
Therefore, the probabilistic matrix $J_n - \clM$ computes $F(\Peq_n)$ with error $\epsilon$.
Since $\rank(F(\Peq_n)) \leq 1 + \rank(T(\Peq_n))$, we have that $\RI(\Peq_n,q) = O(1/\epsilon)$.
\end{proof}

%
%

\section{Acknowledgments}

We thank Srijita Kundu, Swagato Sanyal and Alexander Belov for helpful discussions, and Krišjānis Prūsis for suggestions on the presentation.
We greatly thank Miklos Santha for hospitality during our stay at CQT, and Andris Ambainis for support.
We are grateful to the anonymous reviewers for suggestions and useful feedback, and for pointing out the precise lower bound for prime $q$ in Section \ref{sec:poly}.

\bibliography{hoeteck}

\begin{thebibliography}{GVW15}

\bibitem[AIK06]{AIK04}
Benny Applebaum, Yuval Ishai, and Eyal Kushilevitz.
\newblock Cryptography in {NC}{$^0$}.
\newblock {\em SIAM J. Comput.}, 36(4):845--888, 2006.

\bibitem[App11]{App11}
Benny Applebaum.
\newblock Randomly encoding functions: A new cryptographic paradigm - (invited
  talk).
\newblock In {\em ICITS}, pages 25--31, 2011.

\bibitem[AW17]{AW17}
Josh Alman and Ryan Williams.
\newblock Probabilistic rank and matrix rigidity.
\newblock In {\em Proceedings of the 49th Annual ACM SIGACT Symposium on Theory
  of Computing}, STOC 2017, pages 641--652, New York, NY, USA, 2017. ACM.

\bibitem[BDL13]{BDL13}
Abhishek Bhowmick, Zeev Dvir, and Shachar Lovett.
\newblock New bounds for matching vector families.
\newblock In {\em STOC}, pages 823--832, 2013.

\bibitem[BHR12]{BHR12}
Mihir Bellare, Viet~Tung Hoang, and Phillip Rogaway.
\newblock Foundations of garbled circuits.
\newblock In {\em ACM CCS}, 2012.
\newblock Also Cryptology ePrint Archive, Report 2012/265.

\bibitem[BW07]{BW07}
Dan Boneh and Brent Waters.
\newblock Conjunctive, subset, and range queries on encrypted data.
\newblock In {\em TCC}, pages 535--554, 2007.

\bibitem[DG15]{DG15}
Zeev Dvir and Sivakanth Gopi.
\newblock 2-server {PIR} with sub-polynomial communication.
\newblock In {\em STOC}, pages 577--584, 2015.

\bibitem[DGY11]{DGY11}
Zeev Dvir, Parikshit Gopalan, and Sergey Yekhanin.
\newblock Matching vector codes.
\newblock {\em {SIAM} J. Comput.}, 40(4):1154--1178, 2011.

\bibitem[DH13]{DH13}
Zeev Dvir and Guangda Hu.
\newblock Matching-vector families and {LDC}s over large modulo.
\newblock In {\em Approximation, Randomization, and Combinatorial Optimization.
  Algorithms and Techniques}, pages 513--526, Berlin, Heidelberg, 2013.
  Springer Berlin Heidelberg.

\bibitem[Efr12]{Efremenko12}
Klim Efremenko.
\newblock 3-query locally decodable codes of subexponential length.
\newblock volume~41, pages 1694--1703, 2012.

\bibitem[FKN94]{FKN94}
Uriel Feige, Joe Kilian, and Moni Naor.
\newblock A minimal model for secure computation.
\newblock In {\em STOC}, pages 554--563, 1994.

\bibitem[Gro00]{Grolmusz2000}
Vince Grolmusz.
\newblock Superpolynomial size set-systems with restricted intersections mod 6
  and explicit ramsey graphs.
\newblock {\em Combinatorica}, 20(1):71--86, 2000.

\bibitem[GVW15]{GVW15}
Sergey Gorbunov, Vinod Vaikuntanathan, and Hoeteck Wee.
\newblock Predicate encryption for circuits from {LWE}.
\newblock In {\em CRYPTO (2)}, pages 503--523, 2015.
\newblock Also, Cryptology ePrint Archive, Report 2015/029.

\bibitem[IK00]{IK00}
Yuval Ishai and Eyal Kushilevitz.
\newblock Randomizing polynomials: A new representation with applications to
  round-efficient secure computation.
\newblock In {\em FOCS}, pages 294--304, 2000.

\bibitem[KSW08]{KSW08}
Jonathan Katz, Amit Sahai, and Brent Waters.
\newblock Predicate encryption supporting disjunctions, polynomial equations,
  and inner products.
\newblock In {\em EUROCRYPT}, pages 146--162, 2008.

\bibitem[LV18]{LV18}
Tianren Liu and Vinod Vaikuntanathan.
\newblock Breaking the circuit-size barrier in secret sharing.
\newblock STOC 2018. Cryptology ePrint Archive, Report 2018/333, 2018.

\bibitem[LVW17]{LVW17}
Tianren Liu, Vinod Vaikuntanathan, and Hoeteck Wee.
\newblock Conditional disclosure of secrets via non-linear reconstruction.
\newblock In {\em Advances in Cryptology - {CRYPTO} 2017 - 37th Annual
  International Cryptology Conference, Santa Barbara, CA, USA, August 20-24,
  2017, Proceedings, Part {I}}, pages 758--790, 2017.

\bibitem[LVW18]{LVW18}
Tianren Liu, Vinod Vaikuntanathan, and Hoeteck Wee.
\newblock Towards breaking the exponential barrier for general secret sharing.
\newblock In {\em Advances in Cryptology - {EUROCRYPT} 2018 - 37th Annual
  International Conference on the Theory and Applications of Cryptographic
  Techniques, Tel Aviv, Israel, April 29 - May 3, 2018 Proceedings, Part {I}},
  pages 567--596, 2018.

\bibitem[PS03]{PSbook}
Manoj Prabhakaran and Amit Sahai.
\newblock {\em Secure Multi-Party Computation}.
\newblock IOS Press, 2003.

\bibitem[Yao82]{Yao82}
Andrew Chi-Chih Yao.
\newblock Theory and applications of trapdoor functions.
\newblock In {\em FOCS}, pages 80--91, 1982.

\bibitem[Yek08]{Yek08}
Sergey Yekhanin.
\newblock Towards 3-query locally decodable codes of subexponential length.
\newblock {\em J. {ACM}}, 55(1):1:1--1:16, 2008.

\end{thebibliography}

\end{document}